\documentclass[]{article}
\usepackage[margin=1in]{geometry}
\usepackage{graphicx,amsmath,amsfonts,dsfont,amsthm}
\usepackage{float}
\usepackage{natbib}
\usepackage{algorithm, algorithmic}
\usepackage{stackrel}

\newtheorem{proposition}{Proposition}

\newcommand{\defeq}{:=}
\graphicspath{{fig/}}

\title{Generative AI for Bayesian Computation}

\author{
	\makebox[.4\linewidth]{Nicholas G. Polson\footnote{Nick Polson is a Professor in Econometrics and Statistics at the U Chicago Booth, email: ngp@chicagobooth.edu. We thank the participants of the GDRR 13 meeting in Madrid, May 24-26}}\\
	\textit{Booth School of Business}\\
	\textit{University of Chicago}\\
	\and
	\makebox[.4\linewidth]{Vadim Sokolov\footnote{Vadim Sokolov is an Associate Professor in Operations Research at George Mason University, email: vsokolov@anl.gov}}\\
	\textit{Department of Systems Engineering and Operations Research}\\
	\textit{George Mason University}\\
}

\begin{document}

\date{First Draft: Dec 10, 2022\\
This Draft: February 12,  2024}

\maketitle

\begin{abstract}
\noindent  Bayesian Generative AI (BayesGen-AI)  methods are developed and applied to Bayesian computation.  BayesGen-AI  reconstructs the posterior distribution by directly modeling the parameter of interest as a mapping (a.k.a. deep learner) from a large simulated dataset.
This provides a generator that we can evaluate at the observed data and provide draws from the posterior distribution.
This method applies to all forms of Bayesian inference including parametric models,  likelihood-free models, prediction and maximum expected utility problems. 
Bayesian  computation is then equivalent to high dimensional non-parametric regression.
Bayes Gen-AI main advantage is that it is density-free and therefore provides an alternative to Markov Chain Monte Carlo.  It has a  number of advantages 
over vanilla generative adversarial networks (GAN) and approximate Bayesian computation (ABC) methods due to the fact that the generator is simpler to learn than a GAN architecture and is more flexible than kernel smoothing implicit in ABC methods.
Design of the Network Architecture requires careful selection of features (a.k.a. dimension reduction) and  nonlinear architecture for inference. As a generic architecture, we propose a deep quantile neural network and a uniform base distribution at which to evaluate the generator.
To illustrate our methodology, we provide two real data examples, the first in  traffic flow prediction  and the second in building a surrogate for satellite drag data-set.
Finally, we conclude with  directions for future research.
\end{abstract}

\section{Introduction}\label{sec:intro}
Our goal is to develop Bayesian Generative AI (BayesGen-AI) algorithms. Bayes Gen-AI calculates the posterior distributions and functionals of interest that arise in machine learning task and statistical inference by directly writing the parameter of interest as a stochastic mapping (a.k.a deep learner) and thus avoids the use of densities and  Markov Chain Monte Carlo (MCMC) simulation, by directly learning the inverse posterior mapping from parameters $ \theta $ to data $y$, our method doesn't rely on the use of iterative simulation methods or objective functions that involve density calculations. BayesGen-AI  applies to all forms of Bayesian inference including parametric models,  likelihood-free models, prediction and maximum expected utility problems.  To illustrate our methodology, we use two real world data examples: first, a nonlinear traffic flow prediction problem based on \cite{polson2017deep,polson2015a} and, second, a satellite drag dataset. 

BayesGen-AI starts with a large sample from a joint distribution of observables, parameters and an independent base distribution, denoted by  
$ ( y^{(i)} , \theta^{(i)} , \tau^{(i)}  )_{i=1}^N $ where $N$ is large.  For the data generating process we allow for both parameter likelihood inference $ y^{(i)} | \theta^{(i)} \sim p( y | \theta^{(i)} ) $ and likelihood-free models where $ y^{(i)} = f ( \theta^{(i)} ) $ is a given forward map. 
Samples are generated from the prior
$ \theta^{(i)} \sim \pi( \theta ) $, and  a base distribution  $ \tau^{(i)} \sim p(\tau)$, typically a vector of uniforms.
We also allow for the addition of a latent variable which is stochastic  $ z^{(i)} \sim p(z) $ or deterministic regressor variables where $ z^{(i)} = x $ for some given $x$. Hence, our methodology also applies to conditional density estimation.

When the base distribution is uniform, the inverse Bayes map is simply given by the multivariate inverse posterior cdf where we write $\theta = F^{-1}_{\theta \mid y}(\tau)$, where $\tau\sim p(\tau)$.   Hence, given a training data set,  we can directly find  the mapping $\theta^{(i)} = H\left( y^{(i)},\tau^{(i)}\right)$ by learning 
$H$ as a deep neural network from a large simulated training dataset.  The deep learner requires the specification of a sequence of layers including feature extraction, dimension reduction and nonlinearity to pattern match the training data.  Designing the NN architecture for $H$ is the main challenge in Bayes Gen-AI.   For a generic choice, we propose the use of quantile neural networks (QNNs). Our work builds on \cite{jiang2017} who were the first to propose deep learners for dimension reduction methods and to provide asymptotic theoretical results. We also build on the insight by \cite{dabney2018,ostrovski2018} that implicit quantile neural networks can be used to approximated posterior distributions that arise in decision theory. \cite{dabney2017} also show the connection between the widely used $1$-Wasserstein distance and quantile regression.  

Given an observable data point $y$ and a new base draw $\tau$, the posterior $ p(\theta | y) $ is simply determined from the map
$ \theta = H ( y , \tau ) $ where $ \tau \sim U(0,1) $.  From this we can calculate any quantiles of densities that we need. From this we also show how to calculate any posterior
functionals $ E( f(\theta ) | y)$ by integration the appropriate quantile function using the trapezoidal rule.  We provide an $ O( N^{-4} ) $ approximation bound showing the advantage of quantile methods in estimating functionals.
If the base measure is chosen to be uniform, deep Quantile NNs provide a general framework for  training $H$ and for Bayesian decision-making.
Our method also provides an alternative to invertible NN approaches such as normalizing flows, to Gaussian Process surrogates, and we show how ABC methods can be viewed within our framework.

Architecture design to learn an inverse CDF (quantile function), namely $F^{-1}(\tau,y) = f_\theta (\tau,y)$, will use a kernel embedding by augmenting the input space
to the network. The quantile function is represented as a superposition of two other functions $F^{-1}(\tau,y) = f_\theta (\tau,y) = g(\psi(y)\circ \phi(\tau))$ where $\circ$ is the element-wise multiplication operator. Both functions $g$ and $\psi$ are feed-forward neural networks and $ \phi $ is a cosine embedding transform.
To avoid over-fitting, we use a sufficiently large training dataset, see   \cite{dabney2018} in a reinforcement learning context.
Dimension reduction (a.k.a. feature extraction) will draw on methods used in ABC, see \cite{jiang2017} for an approach using deep learning.

Generative AI also allow the researcher to learn a dimensionality-reduced summary (a.k.a. sufficient) statistics along with a non-linear map \citep{jiang2017,albert2022}. A useful interpretation of the sufficient statistic as a posterior mean, which also allows us to view posterior inputs as one of the inputs to the posterior mean. One can also view a  NN is an approximate nonlinear  Bayes filter to perform such tasks \citep{Mueller1997}.  Our framework provides a natural link for black box methods and stochastic methods, as commonly known in the machine learning literature \citep{bhadra2021,breiman2001}.  Quantile deep neural networks and their ReLu/$\tanh$ counterparts provide a natural architecture. Approximation properties of those networks are discussed in \cite{polson2018posterior}. Dimensionality reduction can be performed using auto-encoders and partial least-squares \citep{polson2021a} due to the result by  \cite{brillinger2012,bhadra2021}, see survey by \cite{blum2013} and kernel embeddings approach discussed by \cite{park2016}. Generative AI circumvents the need for methods like MCMC that require the density evaluations.

ABC methods can also be viewed as a generative approach. Here we show that $H$ can be viewed as a  nearest neighbor network. ABC relies on comparing summary statistic $S$ calculated from the observed data to the summary statistics calculated from the simulated data. \cite{jiang2017} show that a natural choice of $S$ is via the posterior mean. \cite{papamakarios2016} shows how to use mixture density networks \citep{bishop1994} to approximate the posterior for ABC calculations.
In an ABC framework with a parametric exponential family, see \cite{beaumont2002} and \cite{nunes2010} for the optimal choice of summary statistics. A local smoothing version of ABC is given in \cite{jiang2018,bernton2019,fearnhead2012}, \cite{longstaff2001} take a basis function approach. \cite{pastorello2003} provide an estimation procedure when latent variables are present.  

In the statistical and engineering literature, generative models often arise in the context of inverse problems \cite{baker2022} and decision-making \cite{dabney2017}. In the context of inverse problems, the prediction of a mean is rarely is an option, since average of several correct values is not necessarily a correct value and might not even be feasible from the physical point of view. The two main approaches are surrogate-based modeling and approximate Bayes computations (ABC) \cite{park2016,blum2013,beaumont2002}. Surrogate-based modeling \cite{gramacy2020surrogates} is a general approach to solve inverse problems, which relies on the availability of a forward model $y = f(\theta)$, which is a deterministic  function of parameters $\theta$. The forward model is used to generate a large sample of pairs $(y,\theta)$, which is then used to train a surrogate, typically a Gaussian Process, which can be used to calculate the inverse map.  For a recent review of the surrogate-based approach see \cite{baker2022}. There are multiple papers that address different aspects of surrogate-based modeling.

Bayes Gen-AI provides an alternative to non-parametric Gaussian process-based surrogates which heavily  rely on the informational contribution of each sample point and quickly becomes ineffective when faced with significant increases in dimensionality \citep{shan_survey_2010,donoho_high_dimensional_2000}. Further,  the homogeneous GPs models predict poorly \citep{binois_practical_2018}. Unfortunately, the consideration of each input location to handle these heteroskedastic cases result in analytically intractable predictive density and marginal likelihoods \citep{lazaro_gredilla_sparse_2011}. Furthermore, the smoothness assumption made by GP models hinders capturing rapid changes and discontinuities in the input-output relations. Popular attempts to overcome these issues include relying on the selection of kernel functions using prior knowledge about the target process \citep{cortes_rational_2004}; splitting the input space into subregions so that inside each of those smaller subregions the target function is smooth enough and can be approximated with a GP model  \citep{gramacy_bayesian_2008, gramacy_local_2015, chang_fast_2014}; and learning spatial basis functions \citep{bayarri_framework_2007,wilson_fast_2014,higdon_space_2002}.


For low-dimensional $\theta$, the simplest approach is to discretize the parameter space and the data space and use a lookup table to approximate $p(\theta \mid y)$. This is the approach taken by \cite{jiang2017}. However, this approach is not scalable to high-dimensional $\theta$. For practical cases, when the dimension of $\theta$ is high, we can use conditional independence structure present in the data to decompose the joint distribution into a product of lower-dimensional functions \citep{papamakarios2017}. In machine learning literature a few approaches were proposed that rely on such a decomposition \citep{oord2016,germain2015,papamakarios2017}. Most of those approaches use KL divergence as a metric of closeness between the target distribution and the learned distribution.

A natural approach to model the posterior distribution using a neural network is to assume that parameters of a neural network are random variables to use either approximation or MCMC techniques to model the posterior distribution over the parameters \citep{izmailov2021bayesian}. A slightly different approach is to assume that only weights of the last output layer of a neural network are stochastic \citep{wang2022data,schultz2022}. BayesGen-AI provides a natural alternative to these methods.

The paper is outlined as follows. Section \ref{sec:intro} provides a review of the existing literature. Sections \ref{sec:model} describes our GenAI-Bayes Bayesian algorithm. Section  \ref{sec:qnn} describes quantile neural network architectures.
Sections \ref{sec:applications} provides two real world  applications from traffic flow prediction problem and a satellite drag data-set. Finally, Section \ref{sec:discussion} concludes with directions for future research.

\section{Generative AI for Bayesian Computation}\label{sec:model}

In this section we  provide a general description of the Bayes Gen-AI algorithm. Then we describe  architecture design and quantile neural networks. We begin by establishing the following notations that will be used.
\begin{align*}
y=&  \; \; {\rm outcome \; of \; interest} \\
\theta =&  \; \; {\rm parameters} \\
z=&  \; \; {\rm latent \; variables} \\
\tau =&  \; \; {\rm base \; distribution}  
\end{align*}
Let $(y, \theta )  \in \Re^{n + k} $ be observable data and parameters.  Write $ \theta = ( \theta_1 , \ldots , \theta_k ) \in \Re^k $. Let $ \tau \in \Re^K $  be a base measure and $ z \in \Re^D $ be a vector of latent variables.
Suppose $ \theta \sim \pi (\theta ) $, $ z \sim p(z) $ and $ \tau \sim p( \tau ) $. A natural choice will be  $ \tau$ a vector of uniforms and $ K=k $.  In the presence of latent variables,
the likelihood is given by
$ p( y| \theta ) = \int p( y ,  z | \theta ) d z $.  

We allow for the special case where the latent variables are observed regressor variable and $ p( z ) = \delta_{\{ z=x \}} $ where $ \delta $ denotes a Kr\"onecker delta function. In this case, we can find the posterior $ p ( \theta | x , y ) $.  Moreover, we allow for likelihood-free inference where $ y= f( \theta ) $ for some deterministic map. The key property is that we will never have to know the densities to learn the inverse Bayes map.

The goal is to compute the posterior distribution $ p(\theta \mid  y) $.
The underlying assumptions are that $ \theta \sim \pi( \theta ) $ a prior distribution. Our framework allows for many forms of stochastic data generating processes.  The dynamics of the data generating process are such that it is straightforward to simulate from a  so-called forward model or traditional stochastic model, namely
\begin{equation}
	y = f(\theta ) \; \; {\rm or} \; \;  y | \theta \sim p(y| \theta ) . 
\end{equation}
The idea is quite straightforward, if we could perform high dimensional non-parametric regression, we could simulate a large training dataset of observable parameter, data and base pairs, denoted by  $ ( y^{(i)} , \theta ^{(i)}, \tau^{(i)} )_{i=1}^N $.  If latent variables are part of the model specification we simply attach $ z^{(i)} $ to the triple. The inverse Bayes map is then given by
\begin{equation}
	\theta \stackrel{D}{=} H( y , \tau ),	
\end{equation}
where $\tau$ is the vector with elements from the baseline distribution, such as Gaussian, are simulated training data. Here $ \tau $ is a vector of standard uniforms. The function $ H : \Re^k \times \Re^D \rightarrow \Re^k $ is a deep neural network.  The function $ H $ is again trained using the simulated data $ ( y^{(i)} , \theta^{(i)} )_{i=1}^N $, via quantile regression using the simulated data
\[
\theta^{(i)} = H( y^{(i)},\tau^{(i)}),~i=1,\ldots,N.
\]
The main construction is the specification of the network structure implicit in $H$. We will provide guidelines in the following sections.

In many practical applications $y$ is high-dimensional, and we can improve the performance of the deep neural network by using a summary statistic $S(y)$. Here $ S : \Re^N \rightarrow \Re^k $ is a $k$-dimensional sufficient statistic, which is a lower-dimensional representation of $y$.  The summary statistic is a sufficient statistic in the Bayes sense \citep{kolmogorov1942}.  The main idea is to use a deep neural network to approximate the posterior mean $ E_\pi ( \theta | y ) $ as the optimal estimate of $ S(y)$.

Assuming that we have fitted the deep neural network, $H$, to the training data, we can use the estimated inverse map to evaluate at new $y$ and $\tau$ to obtain a set of posterior samples for any new $y$ by interpolation and evaluating the map 
\begin{equation}
	\theta \stackrel{D}{=} H_N( S(y) , \tau ),	
\end{equation}
where $ H_N $ denotes the estimated map.

 The caveat being is to how to choose $H$ and how well the deep neural network interpolates for the new inputs. We also have flexibility in choosing the distribution of $\tau$, for example, we can also for $\tau$ to be a high-dimensional vector of Gaussians, and essentially provide a mixture-Gaussian approximation for the set of posterior. Gen-AI in a simple way is using pattern matching to provide a look-up table for the map from $y$ to $\theta$. Bayesian computation has then being replaced by the optimisation performed by Stochastic Gradient Descent (SGD). Thus, we can sequentially update $H$ using SGD step as new data arrives. MCMC, in comparison, is computationally expensive and needs to be re-run for any new data point. In our examples, we discuss choices of architectures for $H$ and $S$. Specifically, we propose cosine-embedding for transforming $\tau$. 



\paragraph{Bayes Gen-AI Algorithm:}  A necessary condition is the ability to simulate from the parameters, latent variables, and data process.
This generates a (potentially large) triple
$$
\left\{y^{(i)} ,  \theta^{(i)},\tau^{(i)}\right\}_{i=1}^N, 
$$
By construction, the posterior distribution can be characterized by the von Neumann inverse CDF map.   For $ \theta = ( \theta_1 , \ldots , \theta_k )$ we have
$$
\theta  \stackrel{D}{=} F_{\theta | y }^{-1} ( \tau ), \; \; {\rm where} \; \; \tau \sim U(0,1) 
$$
is a vector of standard uniforms.

Given the observed data $ y $ and new base draw $ \tau $, we then simply evaluate the following posterior map 
$$
\theta \stackrel{D}{=} H(  S(y)  , \tau )
$$
When the base draw  $\tau $ is uniform of same dimension as the parameter vector,  the map $H$ is precisely the inverse posterior cdf.
This characterises the posterior distribution $ p( \theta | y ) $ of interest.
As simulation and fitting deep neural networks is computationally cheap, we typically choose $ N $ to be of order $ 10^6$ or more.  

In many cases, we can estimate a summary statistic, $S$ also via deep learning, and we write the inverse cdf as a composite map, and we have to fit 
$$  \theta^{(i)}= H ( S( y^{(i)}  ) , \tau^{(i)} )  \; \; {\rm where} \; \;   F_{\theta | y }^{-1} = H \circ S 
$$
There are many variations in the architecture design and construction. It is useful to interpret $S$ as a summary statistic providing dimension reduction in the $y$ space.

There is flexibility in the choice of base distribution. For example, we can replace the $\tau$ with a different distribution that we can easily sample from. One example is a multivariate Normal, proposed for diffusion models \cite{sohl-dickstein2015}. The dimensionality of the normal can be large. The main insight is that you can solve a high-dimensional least squares problem with nonlinearity using stochastic gradient descent.  Deep quantile NNs provide a natural  candidate of deep learners. \cite{dabney2018implicit} show that implicit quantile networks can be used to for distributional reinforcement learning and has proposed the use of quantile regression as a method for minimizing the 1-Wasserstein in the univariate case when approximating using a mixture of Dirac functions. Figure \ref{fig:qnn} illustrates the different in using a quantile objective versus pure density simulation based on latent variable approach. 

The 1-Wasserstein distance is the $\ell_1$ metric on the inverse  distribution function. It is also known as earth mover's distance and can be calculated using order statistics \citep{levina2001earth}. For quantile functions $F^{-1}_U$ and $F^{-1}_V$ the 1-Wasserstein distance is given by
$$
W_1 ( F^{-1}_U , F^{-1}_V ) = \int_0^1 | F^{-1}_U ( \tau ) - F^{-1}_V ( \tau ) | d \tau.
$$

\cite{weng2019gan} shows Wasserstein GAN networks out-perform vanilla GAN due to the improved quantile metric $q = F^{-1}_U(\tau)$ minimize the expected quantile loss
$$
E_{U}[\rho_{\tau}(u-q)]
$$

Quantile regression can be shown to minimize the 1-Wasserstein metric. A related loss is the quantile divergence,
$$
q(U,V) = \int_0^1 \int_{F^{-1}_U(q)}^{F^{-1}_V(q)} (F_U(\tau)-q)  dq d \tau.
$$


\begin{figure}[t]\label{fig:qnn}
	\centering
	\includegraphics[width=0.8\textwidth]{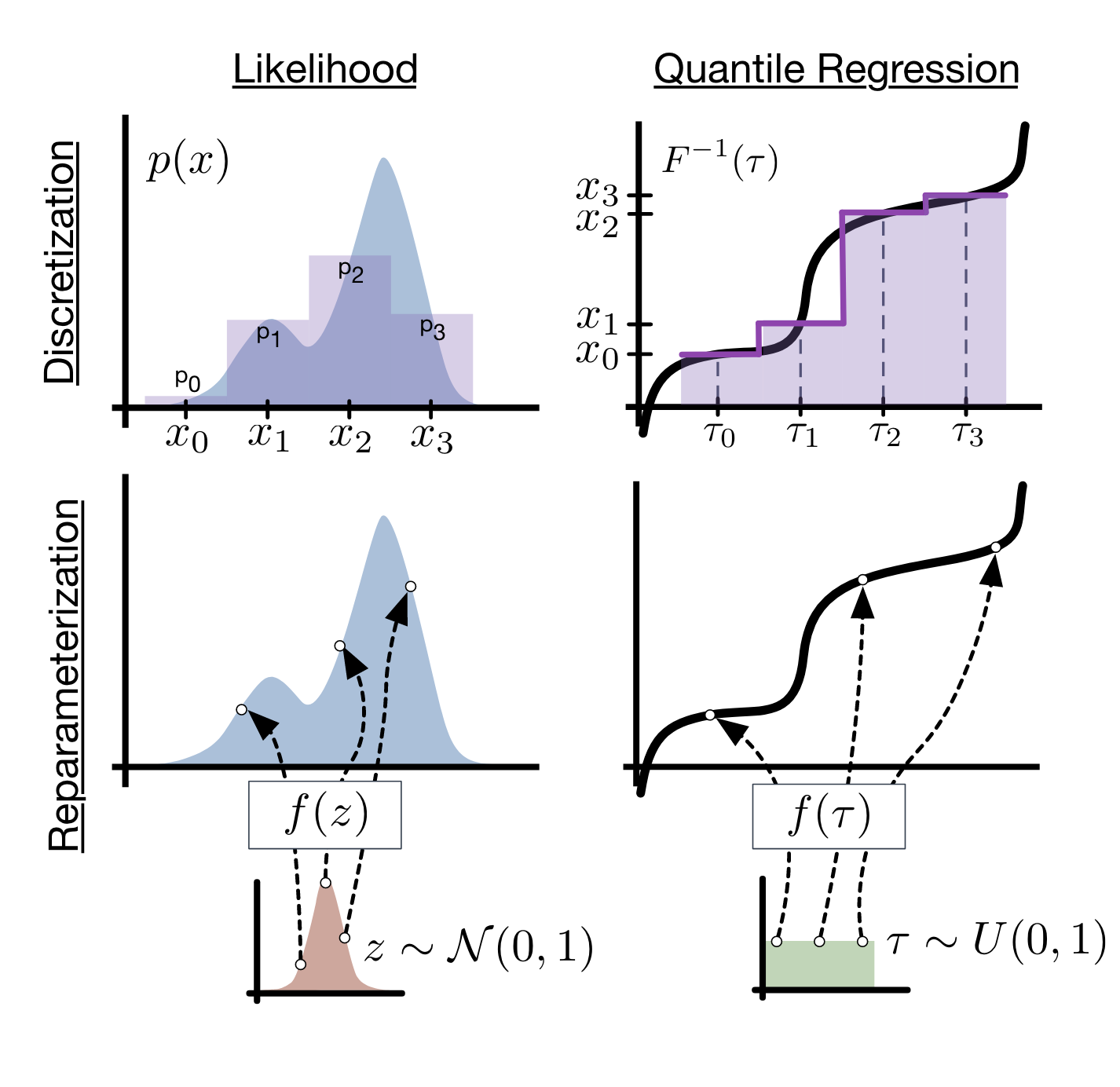}
	\caption{Deep Quantile Neural Network. Source: \cite{ostrovski2018}}
\end{figure}

Hence, even if new $y$ is not in the simulated training dataset  we can still learn the posterior map of interest. The Kolmogorov-Arnold theorem says any multivariate function can be expressed this way.  So in principle if sample size is large enough we can learn the manifold structure in the parameters for any arbitrary nonlinearity. As the dimension of the data $y$ is large, in practice, this requires providing an efficient architecture. 

The following provides a  summary of  Bayes Gen-AI algorithms:

\begin{algorithm}[H]
   \caption{Gen-AI for Bayesian Computation (GenAI-Bayes)}\label{algorithm_a}
\begin{algorithmic}[Gen-AI-Bayes]
   \STATE Simulate $\theta^{(i)} \sim \pi(\theta)$. Simulate $ ( y^{(i)} \mid \theta^{(i)} )_{1 \leq i \leq N}  \sim p(y\mid \theta)$ or $ y^{(i)} = f ( \theta^{(i)} )  $ and  $\tau^{(i)} \sim p(\tau)$ 
   \STATE Train $H$ using the simulated dataset for $ i = 1 , \ldots N $, via $$\theta^{(i)} = H( y^{(i)},\tau^{(i)})$$ 
   \STATE For the observed $y$, calculate a sample from $p(\theta \mid y)$  with a new base draw $ \tau \sim p( \tau ) $ by 
   $$\theta \stackrel{D}{=} H(y , \tau )$$
\end{algorithmic}
\end{algorithm}

In the case when $y$ is high-dimensional, and we can learn a summary statistic $S(y)$, we can use the following algorithm:
\begin{algorithm}[H]
	\caption{GenAI-Bayes with Dimension Reduction}\label{algorithm_ab}
 \begin{algorithmic}[Gen-AI-Bayes]
	\STATE Simulate $\theta^{(i)} \sim \pi(\theta)$. Simulate $ ( y^{(i)} \mid \theta^{(i)} )_{1 \leq i \leq N}  \sim p(y\mid \theta)$ or $ y^{(i)} = f ( \theta^{(i)} )  $ and  $\tau^{(i)} \sim p(\tau)$ 
	\STATE Learn a summary statistic $S(y)$ using a deep neural network
	\STATE Train $H$ using the simulated dataset for $ i = 1 , \ldots N $, via $$\theta^{(i)} = H( S(y^{(i)}),\tau^{(i)})$$ 
	\STATE For the observed $y$, calculate a sample from $p(\theta \mid y)$  with a new base draw $ \tau \sim p( \tau ) $ by 
	$$\theta \stackrel{D}{=} H(S(y) , \tau )$$
 \end{algorithmic}
 \end{algorithm}

Consider now the case where there exist latent variables, $z$,  related to the data density via 
$$ p( y | \theta ) = \int p(y,z | \theta ) d z  $$ 
Again assume that  and when  $ z | \theta $ and $ y | z , \theta $ are  straightforward to simulate either from a density of forward model. Hence, we can simulate  
$ y|\theta $ by composition.

\begin{algorithm}
   \caption{Gen-AI with Latent Variables (GenAI-Bayes-Latent)}\label{algorithm_b}
\begin{algorithmic}[Gen-AI-Bayes]
   \STATE Simulate $\theta^{(i)} \sim \pi(\theta)$. Simulate $y^{(i)} , z^{(i)} \mid \theta^{(i)} \sim p(y,z \mid \theta)$, $i=1,\ldots,N$ or $ y^{(i)} = f( \theta^{(i)} ) $.
   \STATE  Let $S$ denotes feature variables:
   Train $H$ and $S$, using $\theta^{(i)} = H(S( z^{(i)}, y^{(i)}) ,\tau^{(i)})$, where $\tau^{(i)} \sim p(\tau) $
   \STATE For a given $y$, calculate a sample from $p(\theta \mid y)$ with  $\tau$  a new base draw
   $$\theta \stackrel{D}{=} H \left  ( S((z, y) , \tau \right )$$ 
\end{algorithmic}
\end{algorithm}

To learn the feature selection variables $ S(y)$, 
\cite{jiang2017} propose to use a deep learner to approximate the posterior mean $ E_\pi ( \theta | y ) $ as the optimal estimate of $ S(y)$.
Specifically, they construct a minimum squared error estimator $ \hat{\theta} ( y) $ from a large simulated dataset and calculate 
$$
\min_\psi \frac{1}{N} \sum_{i=1}^N \Vert f_\psi ( y^{(i)} ) - \theta^{(i)} \Vert^2_2
$$
where $ f_\psi$ denotes a DNN with parameter $ \psi$.  The resulting estimator $  \hat{\theta } (y ) = f_{ \hat{\psi} } ( y^{(i)} ) $ approximates the feature vector $ S(y ) $.

Together with the following architecture for the summary statistic neural network
\begin{align*}
	H^{(1)} = & \tanh\left(W^{(0)}H^{(0)}+b^{(0)}\right)\\
	H^{(2)} = & \tanh\left(W^{(1)}H^{(1)}+b^{(1)}\right)\\
	& \vdots\\
	H^{(L)} = & \tanh\left(W^{(L-1)}H^{(L-1)}+b^{(L-1)}\right)\\
	\hat\theta = & W^{(L)}H^{(L)}+b^{(L)},
\end{align*}
where $H^{(0)} = \theta$ is the input, and $\hat \theta$ is the summary statistic output. ReLU activation function can be used instead of $\tanh$.
In our applications we use ReLU.

There is still the question of approximation and the interpolation properties of a DNN.

\paragraph{Folklore of Deep Learning:}  \textit{Shallow Deep Learners provide good representations of multivariate functions and are good interpolators}.

\vspace{0.1IN}

A number of authors have provided theoretical approximation results that apply to BayesGen-AI. For example, interpolation properties of quantile neural networks were recently studied by \cite{padilla2022quantile} and \cite{shen2021deep}, \cite{schmidt2020nonparametric}


\subsection{Dimension Reduction}

Given $y$, the posterior density is denoted by $ p(\theta \mid  y ) $.  Here $ y = ( y_1 , \ldots , y_n ) $ is high
dimensional. Moreover, we need the set of  posterior probabilities $   \pi_{\theta \mid y}(\theta \in B\mid y) $ for all Borel sets $B$.
Hence, we need two things, dimension reduction for $y$. 
The whole idea is to estimate "maps" (a.k.a. transformations/feature extraction) of the output data $y$, thus, it is reduced to uniformity.

There is a nice connection between the posterior mean and the sufficient statistics, especially minimal sufficient statistics in the exponential family. If there exists a sufficient statistic $S^*$ for $\theta$, then \cite{kolmogorov1942} shows that for almost every $y$, $\pi(\theta\mid y) = \pi(\theta\mid S^*(y))$, and further  $S(y) = E_{\pi}(\theta \mid y) = E_{\pi}(\theta \mid S^*(y))$ is a function of $S^*(y)$. In the special case of an exponential family with minimal sufficient statistic $S^*$ and parameter $\theta$, the posterior mean $S(y) = E_{\pi}(\theta \mid y)$ is a one-to-one function of $S^*(y)$, and thus is a minimal sufficient statistic.

Deep learners are good interpolators (one of the folklore theorems of machine learning). Hence, the set of posteriors $ p(\theta \mid  y ) $ is characterized by the distributional identity
$$
\theta  \stackrel{D}{=} H ( S(y) , \tau_k ),  \; \; {\rm where}  \;  \; \tau_k \sim U(0,1 ),
$$
where $y = ( y_1 , \ldots , y_n )  $ and $ \tau_k $ is a vector of standard uniforms.

\paragraph{Summary Statistic:}  Let $S(y)$ is sufficient summary statistic in the Bayes sense \citep{kolmogorov1942}, if for every prior $\pi$
\[
f_B (y) \defeq    \pi_{\theta \mid y}(\theta \in B\mid y) = \pi_{\theta \mid s(y)}(\theta \in B\mid s(y)).
\]
Then we need to use our pattern matching dataset $ ( y^{(i)} , \theta^{(i)} ) $ which is simulated from the prior and forward model to "train" the set of functions
$ f_B (y) $, where we pick the sets $ B = ( - \infty , q ] $ for a quantile $q$. Hence, we can then interpolate the mapping to the observed data.

The notion of a summary statistic is prevalent in the ABC literature and is tightly related to the notion of  a Bayesian 
sufficient statistic $S^*$ for $\theta$, then (Kolmogorov 1942), for almost every $y$,
$$\pi(\theta \mid  Y=y) = \pi(\theta \mid S^*(Y) = S^*(y))$$
Furthermore,  $S(y) = \mathrm{E}\left(\theta \mid Y = y\right) = \mathrm{E}_{\pi}\left(\theta \mid S^*(Y) = S^*(y)\right)$ is a function of $S^*(y)$. In the case of exponential family, we have $S(Y) =  \mathrm{E}_{\pi}\left(\theta | Y \right)$ is a one-to-one function of $S^*(Y)$, and thus is a minimal sufficient statistic.

Sufficient statistics are generally kept for parametric exponential families, where $S(\cdot)$ is given by the specification of the probabilistic model.
However, many forward models have an implicit likelihood and no such structures. The generalization of sufficiency is a summary statistics 
(a.k.a. feature extraction/selection in a neural network).  Hence, we make the assumption that there exists a set of features such that the dimensionality of the problem is reduced

Learning $S$ can be achieved in a number of ways. First, $S$ is of fixed dimension $ S \in \Re^s $ even though $y= ( y_1 , \ldots y_n )$.
Typical architectures include Auto-encoders and traditional dimension reduction methods.  \cite{polson2021a} propose to use a theoretical result of Brillinger methods to perform a linear mapping $S(y) = W y$ and learn $W$ using PLS.  \cite{nareklishvili2022a} extend this to IV regression and casual
inference problems.

\paragraph{PLS}
Another architecture for finding summary statistics is to use PLS. Given the parameters and data, the map is 
\[
\theta^{(i)} = H\left( B y^{(i)} ,\epsilon\right)
\]
We can find a set of linear maps $S(y) = By$. This rule also provide dimension reduction. Moreover, due to orthogonality of $y,\epsilon$, we can simply consistently estimate $\hat S(y) = \hat By$  via $\theta^{(i)} = H(\hat By)$. A key result of \cite{brillinger2012} shows that we can use linear SGD methods and partial least squares to find $\hat B$. 

\paragraph{ABC}
Approximate Bayesian Computations (ABC) is a common approach in cases when likelihood is not available, but samples can be generated from some model, e.g. epidemiological simulator. The ABC rely on comparing summary statistic calculated from data $s_y$ and of from observed output $s(\theta)$ and approximates
\[
p(\theta,s(\theta)\mid s_y,\epsilon) \propto \pi(\theta)f(s(\theta)\mid \theta) K_{\epsilon}\left(\|s(\theta) - s_y\|\right).
\] 
Then the approximation to the posterior is simply
$p(\theta \mid s_y) = \int p_\epsilon (\theta,s(\theta)\mid s_y) ds $.
Then the ABC algorithm simply samples from $p(\theta)$, then generates  summary statistic $s(\theta)$ and rejects the sample with probability proportional to $ K_{\epsilon}(\|s(\theta) - s_y\|)$. The Kernel function $K$ can be a simple indicator function
\[
	K_{\epsilon}(\|s(\theta) - s_y\|) = \mathds{1}(\|s(\theta) - s_y\| < \epsilon)
\]
The use of  deep neural networks to select $s(\theta)$ has been proposed by 
\cite{jiang2017}.

Conventional ABC methods suffers from the main drawback that the samples do not come from the true posterior, but an approximate one, based on the $\epsilon$-ball approximation of the likelihood, which is  a non-parametric local smoother. Theoretically, as $\epsilon$ goes to zero, you can guarantee samples from the true posterior. However, the number of sample required is prohibitive. Our method circumvents this by replacing the ball with a deep learning generator and directly model the relationship between the posterior and a baseline uniform Gaussian distribution. Our method is also not a density approximation, as many authors have proposed. Rather, we directly use $L_2$ methods and Stochastic Gradient Descent to find transport map from $\theta$ to a uniform or a Gaussian random variable. The equivalent to the mixture of Gaussian approximation is to assume that our baseline distribution is high-dimensional Gaussian. Such models are called the diffusion models in literature. Full Bayesian computations can then be reduced to high-dimensional $L_2$ optimization problems with a carefully chosen neural network.

\paragraph{Invertible NNs} In some cases, we can model $ H : \Re^n \rightarrow \Re^n $ as an invertible map where $ x = g(z) $  is generated from a base density $p(z)$. 
When $g$ is chosen to be an invertible neural network with a structured diagonal Jacobian that is easy to compute we can calculate densities. This makes the mapping $g$ far easier to learn.
Our objective function can then be standard log-likelihood based on 
$$
p(x ) = p( z ) | {\rm det} J_z |^{-1}, \; \;  {\rm where} \; \; J_z = \frac{\partial g(z) }{\partial z } 
$$
Bayes Flow provides an example of this type of architecture.

The second class of methods proposed on machine leaning literature involves using deep learners to approximate an inverse CDF function or a more general approach that represents the target distribution over $\theta$ as a marginalization over a nuance random variable $z$ \citep{kingma2022}. In the case of inverse CDF, the latent variable $z$ is simply uniform on $(0,1)$ \citep{bond-taylor2022}. One of the approaches of this type is called Normalizing flows. Normalizing flows provide an alternative approach of defining a deterministic map $\theta \mid x, \theta_g  = G(z,x, \theta_g)$ that transforms a univariate random variable $z\sim p(z)$ to a sample from the target distribution $G(z,x, \theta_g)= x \sim F(\theta)$. If transformation $G$ is invertible ($G^{-1}$ exists) and differentiable, then the relation between target density $F$ and the latent density $p(z)$ is given by \cite{Rezende15}:
\begin{equation}
    F(\theta) =p(z)\left|\det\frac{\partial G^{-1}}{\partial z}\right| = p(z)\left|\det\frac{\partial G}{\partial z}\right|^{-1},
\end{equation}
where $z = G^{-1}(y)$. A typical procedure for estimating the parameters of the map $G$ relies on maximizing the log-likelihood
\begin{equation}
\label{eq:max_likelihood}
  \log p(z) + \log \left|\det\frac{\partial G^{-1}}{\partial z}\right|
\end{equation}
The normalizing flow model requires constructing map $G$ that have tractable inverse and Jacobian determinant. It is achieved by representing $G$ as a composite map
\begin{equation}
\label{eq:NF1}
G = T_k \circ \cdots \circ T_1,
\end{equation}
and to use simple building block transformations $T_i$ that have tractable inverse and Jacobian determinant. 

We now describe our  Quantile Neural Network approach. 

\section{Quantile Neural Networks}\label{sec:qnn}

\subsection{Learning Quantiles}

The quantile regression likelihood function as an asymmetric function that penalizes overestimation errors with weight $\tau$ and underestimation errors with weight $1-\tau$. For a given input-output pair $(x,y)$, and the quantile function $f(x,\theta)$, parametrized by $\theta$, the quantile loss is $\rho_{\tau}(u) = u(\tau-I(u<0))$, where $u = y-f(x)$. From the implementation point of view, a more convenient form of this function is 
\[
\rho_{\tau}(u) = \max(u\tau,u(\tau-1)).
\]
Given  a training data $\{x_i,y_i\}_{i=1}^N$, and given quantile $\tau$, the loss is 
\[
L_{\tau}(\theta) = 	\sum_{i=1}^{N} \rho_{\tau}(y_i - f(\tau,x_i,\theta)).
\]
Further, we empirically found that adding a means-squared loss to this objective function, improves the predictive power of the model, thus the loss function, we use is 
\[
	\alpha L_{\tau}(\theta) + \frac{1}{N} \sum_{i=1}^{N} (y_i - f(x_i,\theta))^2.
\]
One approach to learn the quantile function is to use a set of quantiles $0<\tau_1<\tau_2,\ldots,\tau_K<1$ and then learn $K$ quantile functions simultaneously by minimizing
\[
L(\theta,\tau_1,\ldots,\tau_K) = 	\frac{1}{N K}\sum_{i=1}^{N}\sum_{k=1}^{K}\rho_{\tau_k}(y_i - f_{\tau_k}(x_i,\theta_k)).
\]
The corresponding optimization problem of minimizing $L(\theta)$ can be augmented by adding a non-crossing constraint
\[
	f_{\tau_i}(x,\theta_i) < f_{\tau_j}(x,\theta_j), ~ \forall X, ~ i<j.
\]
The non-crossing constraint has been considered by several authors, including \cite{chernozhukov2010,cannon2018}.

\paragraph{Cosine Embedding for $\tau$}

To learn an inverse CDF (quantile function) $F^{-1}(\tau,y) = f_\theta (\tau,y)$ we will use a kernel embedding trick and augment the predictor space. We then represent the quantile function is a function of superposition for two other functions $$F^{-1}(\tau,y) = f_\theta (\tau,y) = g(\psi(y)\circ \phi(\tau))$$ where $\circ$ is the element-wise multiplication operator. Both functions $g$ and $\psi$ are feed-forward neural networks. $ \phi $ is a cosine embedding.
To avoid over-fitting, we use a sufficiently large training dataset, see   \cite{dabney2018} in a reinforcement learning context.

Let $g$ and $\psi$ be feed-forward neural networks and $\phi$ a cosine embedding given by
\[
	\phi_j(\tau) = \mathrm{ReLU}\left(\sum_{i=0}^{n-1}\cos(\pi i \tau)w_{ij} + b_j\right).
\]
We now illustrate our approach with  a simple synthetic dataset.

\subsection{Synthetic Data}
Consider a synthetic data generated from the model 
\begin{align*}
x &\sim U(-1,1) \\
y &\sim N(\sin(\pi x)/(\pi x), \exp(1-x)/10).
\end{align*}
The true quantile function is given by 
\[
f_{\tau}(x) = \sin(\pi x)/(\pi x) + \Phi^{-1}(\tau)\sqrt{ \exp(1-x)/10},
\]

\begin{figure}[H]
\centering
\begin{tabular}{cc}
\includegraphics[width=0.45\linewidth]{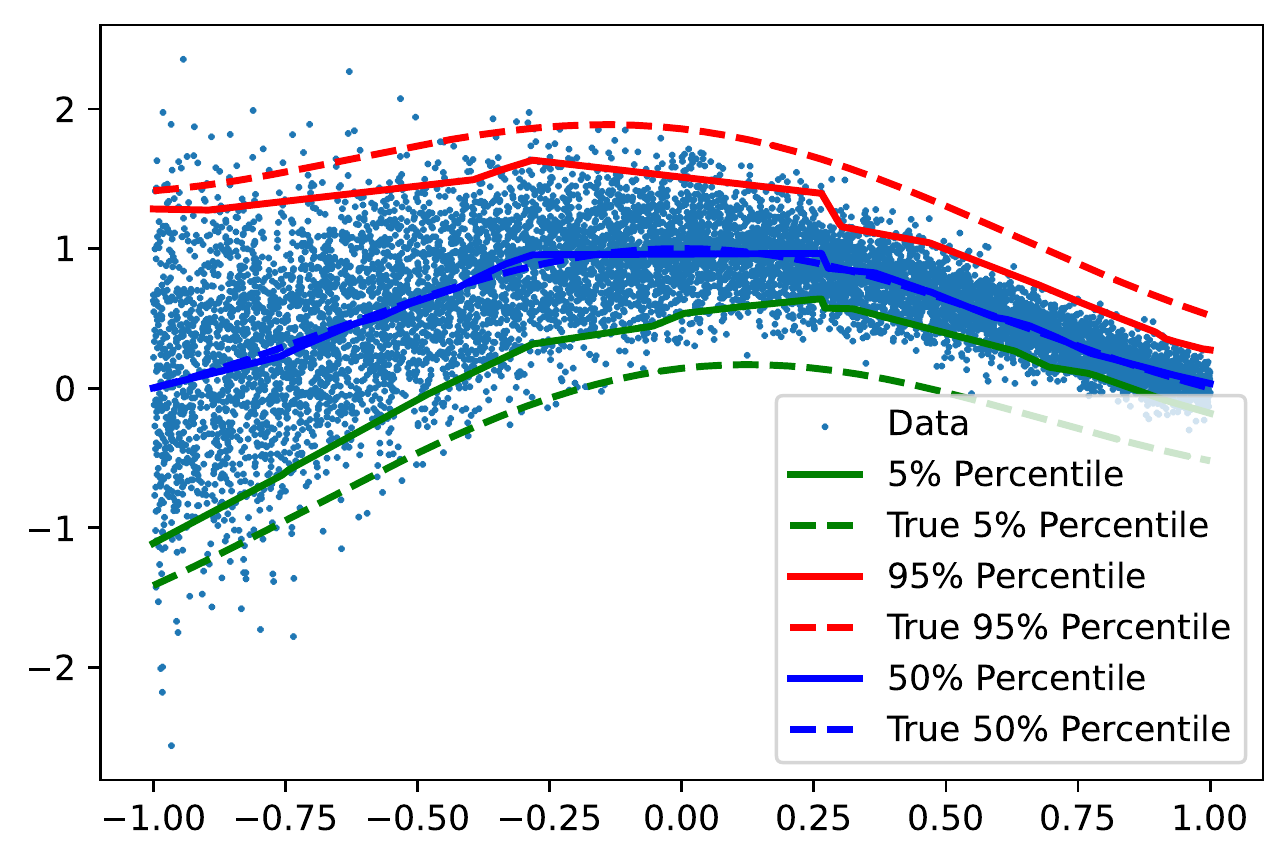} & \includegraphics[width=0.45\linewidth]{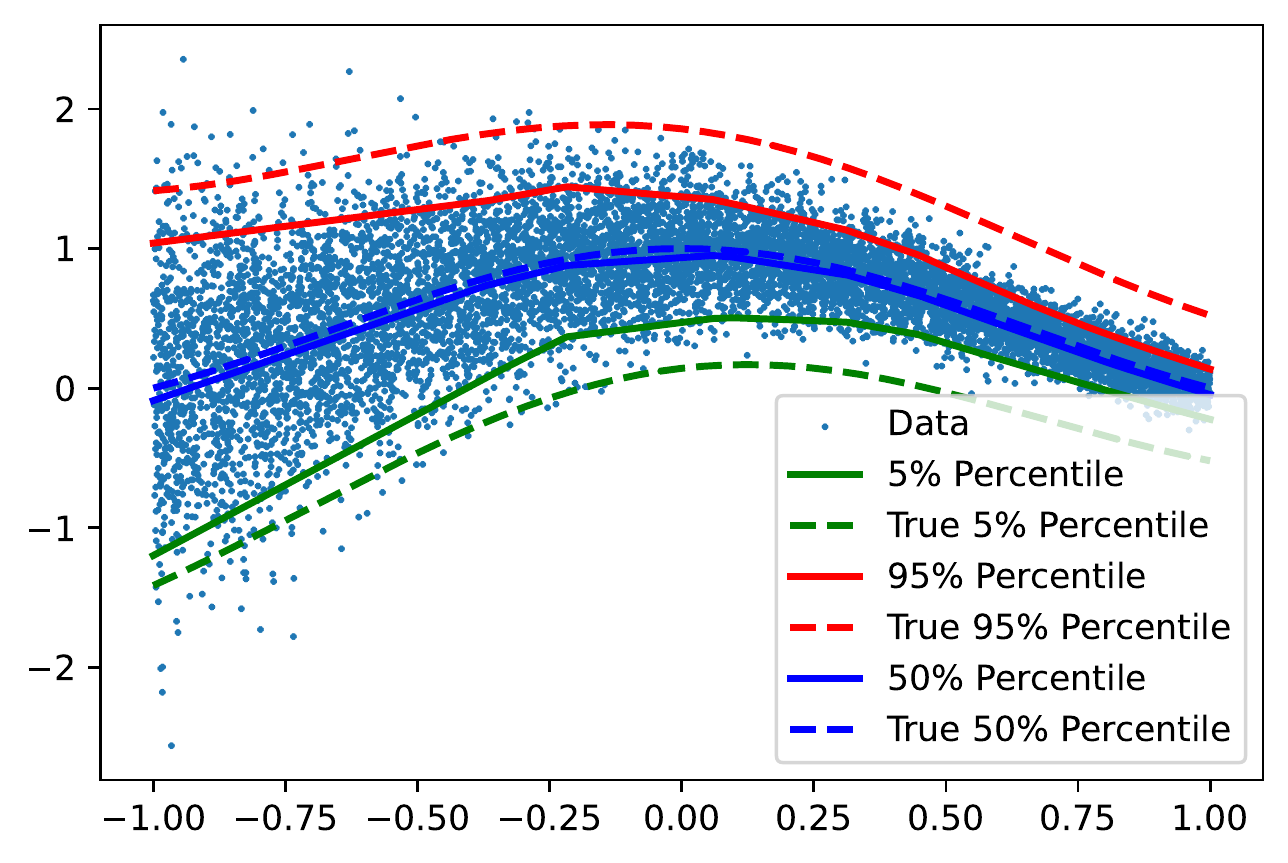}\\
(a) Implicit Quantile Network & (b) Explicit Quantile Network
\end{tabular}
\caption{We trained both implicit and explicit networks on the synthetic data set. The explicit network was trained for three fixed quantiles (0.05,0.5,0.95). We see no empirical difference between the two}
\label{fig:synthetic}
\end{figure}

We train two quantile networks, one implicit and one explicit. The explicit network is trained for three fixed quantiles (0.05,0.5,0.95). Figure \ref{fig:synthetic} shows fits by both of the networks, we see no empirical difference between the two. 

\subsection{Quantiles as Deep Learners}

One can show that quantile models are direct alternatives to other Bayes computations. Specifically, given $F(y)$, a non-decreasing and continuous from right function.  We define 
$$Q_{\theta| y} (u) \defeq  F^{-1}_{\theta|y}  ( u ) = \inf \left ( y : F_{\theta|y} (y) \geq u \right ) $$ which is non-decreasing, continuous from left.
Now let  $ g(y)$ to be a non-decreasing and continuous from left  with 
$$ 
g^{-1} (z ) = \sup \left ( y : g(y ) \leq z \right ) 
$$
Then,  the transformed quantile has a compositional nature, namely 
$$
Q_{ g(Y) } ( u ) = g ( Q (u ))
$$
Hence, quantiles act as  superposition (a.k.a. deep Learner). 

This is best illustrated in the Bayes learning model.  We have the following result updating prior to posterior quantiles known as the conditional quantile representation
$$
Q_{ \theta | Y=y } ( u ) = Q_Y ( s )  \; \; 
{\rm where} \; \;   s = Q_{ F(\theta) |  Y=y } ( u) 
$$  
To compute $s$ we use
$$
u = F_{ F(\theta ) | Y=y} ( s  ) = P( F (\theta ) \leq s | Y=y ) 
= P( \theta \leq Q_\theta (s ) | Y=y )  = F_{ \theta | Y=y } ( Q_\theta ( s ) ) 
$$
\cite{parzen2004}  also shows the following probabilistic property of quantiles
$$
\theta = Q_\theta ( F_\theta (\theta ) ) 
$$
Hence, we can increase the efficiency by ordering the samples of $\theta$ and the baseline distribution as the mapping being the inverse CDF is monotonic. 
\subsection{Normal-Normal Learning}
Consider the normal learning model
$y| \theta \sim N ( \theta , \sigma^2 ) \; \; {\rm and} \; \; \theta \sim N( \mu , \tau^2 ) $.
 Given observations $ (y_1 , \ldots , y_n ) $, we can reduce the likelihood using sufficiency to give
 $$
 \bar{y} | \theta , \sigma^2 \sim N \left  ( \theta , \frac{ \sigma^2 }{n} \right)  \; \; {\rm and} \; \;  \theta \sim N( \mu , \tau^2 ) 
 $$
 The posterior mean is
 $$
 \hat{\theta } ( y) = E( \theta | y ) = \frac{\sigma^2 / n }{  \sigma^2 / n + \tau^2 } \mu +  \frac{ \tau^2 }{  \sigma^2 / n + \tau^2 } \bar{y}
 $$
 The prior to posterior quantile map is given by
 $$
 1 - F_{\theta | y } (\theta ) = g( 1 - F_\theta (\theta ) ) 
 $$
 where $ g : [0,1] \rightarrow [0,1] $ is a concentration (aka distortion) function given by
$$
g(x) = \Phi \left (  w \Phi^{-1} ( x ) + b \right )  \; \; , \; w = \tau \sqrt{ \frac{1}{  \sigma^2 / n + \tau^2 } } ,  b = \frac{1}{ \sigma / \sqrt{n} }  \sqrt{ \frac{\tau^2}{\sigma^2/n + \tau^2 } } 
$$
where $\Phi$ is standard normal CDF function. 

\subsection{Calculating Functionals of Interest}

Another advantage of quantile methods is that we can find tight error bounds of $ O( N^{-4} ) $ for estimating functionals of interest.
Another important property of quantiles is that they can be used to calculate functionals of interest
$$
E_{\theta | y } ( \theta ) = \int_{- \infty}^\infty \theta \pi( \theta | y ) d \theta = \int_0^1 Q_{\theta | y } ( \tau ) d \tau 
$$
We can then use the trapezoidal rule and obtain an efficient approximation due lemma of \cite{yakowitz1978weighted}.
\begin{proposition}
Assume $Q(x)$ is a function with a continuous second derivative on $[0, 1]$ and 
$$\theta \equiv \int_0^1 Q(x)dx. $$ 
If $Y_0 \equiv 0, Y_{n+1} \equiv 1$, and $\{Y_i\}_{i=1}^n$ is the ordered sample associated with n independent uniform observations on [0, 1], then for the estimator 
$$
\theta_n = \frac{1}{2} \left[ \sum_{i=0}^n (Q(Y_i) + Q(Y_{i+1})) (Y_{i+1} - Y_i)\right]
$$
We have that for some constant M,
$$
E[(\theta_n - \theta)^2] \leq M/n^4, \qquad \text{for all } n \geq 1. 
$$
\end{proposition}
\begin{proof}
For any particular sample $Y_1, Y_2, ... , Y_n$, we have 
$$
\theta - \theta_n = \sum_{j=1}^n \left( \int_{Y_j}^{Y_{j+1}} Q(t)dt - \frac{1}{2} [(Q(Y_j) + Q(Y_{j+1})) (Y_{j+1} - Y_j)] \right).
$$
A well-know error bound for the trapezoidal rule is that for any numbers a, b such that $0 \leq a < b \leq 1$,
$$
\int_a^b Q(t)dt - \frac{b-1}{2} [Q(b) + Q(a)] = -\frac{(b-a)^3}{12} Q''(\xi),
$$
where $a \leq \xi \leq b$. Using this result we have that if $C \geq |Q''(x)|,$ $0 \leq x \leq 1$, then
\begin{eqnarray*}
|\theta - \theta_n| &\leq& \left| \int_0^{Y_1} Q(t)dt - \frac{1}{2} Y_1 (Q(0) + Q(1)) \right|\\
& & +\left| \left[ \sum_{i=1}^{n-1} \int_{Y_i}^{Y_{i+1}} Q(t)dt -  \frac{1}{2} [(Q(Y_i) + Q(Y_{i+1})) (Y_{i+1} - Y_i)\right] \right|\\
& & +\left| \int_{Y_n}^1 Q(t)dt -  \frac{1}{2} [(Q(Y_n) + Q(Y_1)) (1 - Y_n) \right|\\
&\leq& \frac{C}{12} \sum_{i=1}^{n+1} Z_i^3,
\end{eqnarray*}
where the $Z_i$'s are as defined as follows.

Let $\{Y_i\}_{i=1}^n$ be the ordered sample constructed from $n$ independent, uniform observations on the unit interval. Define
$$
Z_1 = Y_1,\qquad Z_j = Y_j-Y_{j-1},\qquad 2 \leq j \leq n, \text{ and } Z_{n+1} = 1-Y_n
$$

\cite{yakowitz1978weighted} provides the following result.
Let $\{Y_i\}_{i=1}^n$ be the ordered sample constructed from $n$ independent, uniform observations on the unit interval, then
$$
E[Z_i^6] = 6!n!/(n+6)!, \qquad i = 1, 2, ... , n+1,
$$ 
$$
E[Z_i^3Z_j^3] = (3!)^2n!/(n+1)!, \qquad i, j = 1, 2 , ... , n+1,\qquad i\neq j. 
$$

Consequently, letting $i$ and $j$ range from $1$ to $n + 1$, we have
\[
E[(\theta_n - \theta)^2] \leq E\left[ \left( \frac{C}{12} \sum_i Z_i^3 \right)^2 \right] =  \frac{C^2}{144} \left[ \sum_i E[Z_i^6] + \sum_{i \neq j} E[Z_i^3Z_j^3] \right]= O\left(n^{-4}\right).
\]

\end{proof}

\section{Applications} \label{sec:applications}

\subsection{Traffic Data}
We further illustrate our methodology, using data from a sensor on interstate highway  I-55. The sensor is located eight miles from the Chicago downtown on I-55 north bound (near Cicero Ave), which is part of a route used by many morning commuters to travel from southwest suburbs to the city. As shown on Figure \ref{fig:traffic}, the sensor is located 840 meters downstream of an off-ramp and 970 meters upstream from an on-ramp.
\begin{figure}[H]
\centering
\includegraphics[width=0.6\linewidth]{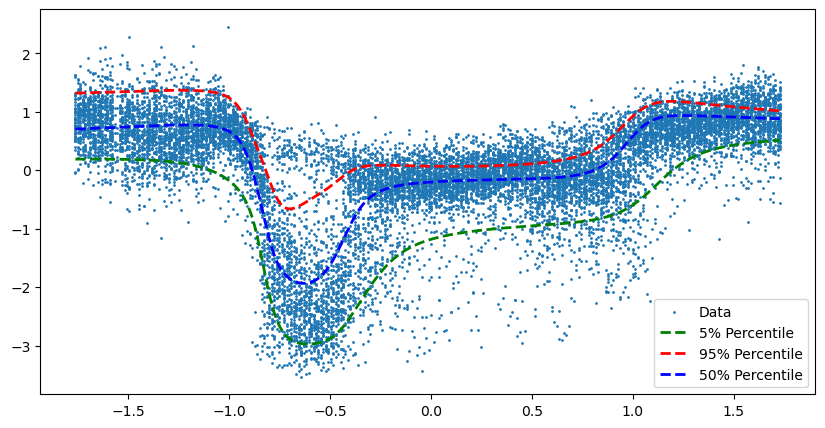}
\caption{Implicit neural network for traffic speed observed on I-55 north-bound towards Chicago}
\label{fig:traffic}
\end{figure} 

In a typical day traffic flow pattern on Chicago's I-55 highway, where sudden breakdowns are followed by a recovery to free flow regime. We can see a breakdown in traffic flow speed during the morning peak period followed by speed recovery. The free flow regimes are usually of little interest to traffic managers. We also, see that variance is low during the free flow regime and high during the breakdown and recovery regimes.

We use the following architecture to model the implicit quantile function
\begin{align*}
	\tau_1 = &\mathrm{ReLU} \left(w^{(1)}_0 + \sum_{i=1}^{64} w^{(1)}_i \cos(i\pi\tau)\right)\\
	x_1 = & \mathrm{ReLU}\left(w^{(2)}_0 + \sum_{i=1}^{64} w^{(2)}_i x_i\right)\\
	z = & \mathrm{ReLU}\left(w^{(3)}_0 + \sum_{i=1}^{64} w^{(3)}_i \tau_ix_i\right)\\
	y_1 = &\mathrm{ReLU}\left(w^{(4)}_0 + \sum_{i=1}^{32} w^{(4)}_i x_i\right)\\
	\hat y = &w^{(5)}_0 + w^{(5)}_i x_i.\\
\end{align*}
We now consider a satellite drag example.

\subsection{Satellite Drag}
Accurate estimation of satellite drag coefficients in low Earth orbit (LEO) is vital for various purposes such as precise positioning (e.g., to plan maneuvering and determine visibility) and collision avoidance. 

Recently, 38 out of 49 Starlink satellites launched by SpaceX on Feb 3, 2022,  experienced an early atmospheric re-entry caused by unexpectedly elevated atmospheric drag, an estimated \$100MM loss in assets. The launch of the SpaceX Starlink satellites coincided with a geomagnetic storm, which heightened the density of Earth's ionosphere. This, in turn, led to an elevated drag coefficient for the satellites, ultimately causing the majority of the cluster to re-enter the atmosphere and burn up. This recent accident shows the importance of accurate estimation of drag coefficients in commercial and scientific applications \citep{Berger23}.

Accurate determination of drag coefficients is crucial for assessing and maintaining orbital dynamics by accounting for the drag force. Atmospheric drag is the primary source of uncertainty for objects in LEO. This uncertainty arises partially due to inadequate modeling of the interaction between the satellite and the atmosphere. Drag is influenced by various factors, including geometry, orientation, ambient and surface temperatures, and atmospheric chemical composition, all of which are dependent on the satellite's position (latitude, longitude, and altitude). 

Los Alamos National Laboratory developed the Test Particle Monte Carlo simulator to predict the movement of satellites in low earth orbit \citep{mehta2014modeling}. The simulator takes two inputs, the geometry of the satellite, given by the mesh approximation and seven parameters, which we list in Table \ref{tab:satellite} below. The simulator takes about a minute to run one scenario, and we use a dataset of one million scenarios for the Hubble Space Telescope \citep{gramacy2020surrogates}. The simulator outputs estimates of the drag coefficient based on these inputs, while considering uncertainties associated with atmospheric and gas-surface interaction models (GSI). 

\begin{table}[H]\label{tab:satellite}
\center
\begin{tabular}{l|c}
Parameter                             & Range            \\\hline
velocity {[}m/s{]}                    & {[}5500, 9500{]} \\
surface temperature {[}K{]}           & {[}100, 500{]}   \\
atmospheric temperature {[}K{]}       & {[}200, 2000{]}  \\
yaw {[}radians{]}                     & $[-\pi,\pi]$     \\
pitch {[}radians{]}                   & $[-\pi/2,\pi/2]$ \\
normal energy AC {[}unitless{]}       & {[}0,1{]}        \\
tangential momentum AC {[}unitless{]} & {[}0,1{]}       
\end{tabular}
\caption{Input parameters for the satellite drag simulator.}
\end{table}

We use the data set of 1 million simulation runs provided by \cite{sun2019emulating}. The data set has 1 million observations, and we use 20\% for training and 80\% for testing out-of-sample performance. The model architecture is given below. We use the Adam optimizer and a batch size of 2048, and train the model for 200 epochs. 

\cite{sauer2023non} provides a survey of modern  Gaussian Process based models for prediction and uncertainty quantification tasks. They compare five different models, and apply them to the same Hubble data set we use in this section. We use two metrics to assess the quality of the model, namely RMSE, which captures predictive accuracy, and continuous rank probability score (CRPS; \cite{gneiting2007strictly,zamo2018estimation}). Essentially CRPS is the absolute difference between the predicted and observed cumulative distribution function (CDF). We use the degenerative distribution with the entire mass on the observed value (Dirac delta) as to get the observed CDF. The lower  CRPSis, the  better.

Their best performing model is treed-GP has the RMSE of 0.08 and CRPS of 0.04, the worst performing model is the deep GP with approximate ``doubly stochastic'' variational inference  has RMSE of 0.23 and CRPS of 0.16. The best performing model in our experiments is the quantile neural network with RMSE of 0.098 and CRPS of 0.05, which is comparable to the top performer from the survey.

Figure \ref{fig:satellite} plots of the out-of-sample predictions  for forty randomly selected responses (green crosses) and compares those to 50th quantile predictions (orange line) and 95\% credible prediction intervals (grey region). 
\begin{figure}[H]\label{fig:satellite}
	\includegraphics[width=1\textwidth]{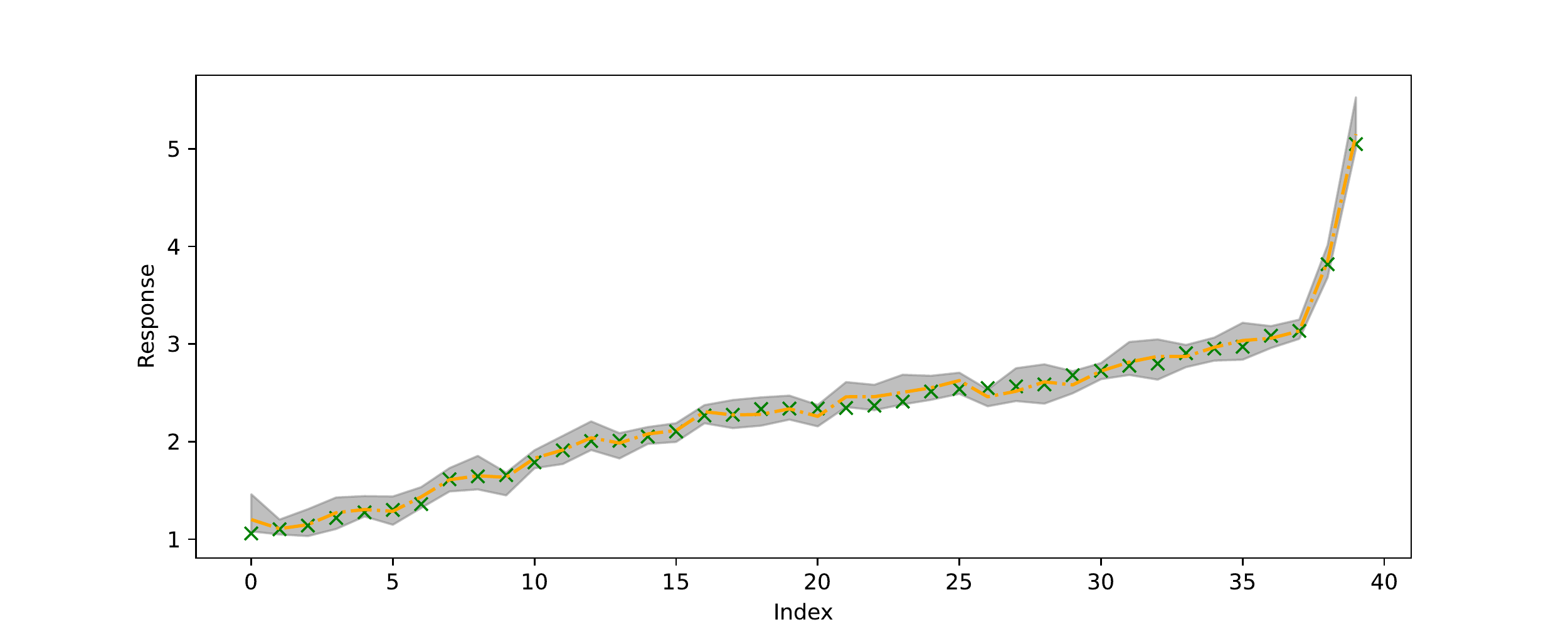}
\caption{Randomly selected 40 out-of-sample observations from the satellite drag dataset. Green crosses are observed values, orange line is predicted 50\% quantile and the grey region is the 95\% credible prediction interval.}
\end{figure}

Figure \ref{fig:satellite-err}(a) compares the out-of-sample predictions (50\% quantiles) $\hat y$ and observed drag coefficients $y$. We can see that histogram resembles a normal distribution centered at zero, with some ``heaviness'' on the left tail, meaning that for some observations, our model under-estimates. The scatterplot in Figure \ref{fig:satellite-err}(b) shows that the model is more accurate for smaller values of $y$ and less accurate for larger values of $y$ and values of $y$ at around three.
\begin{figure}[H]\label{fig:satellite-err}
\begin{tabular}{cc}
	\includegraphics[width=0.5\textwidth]{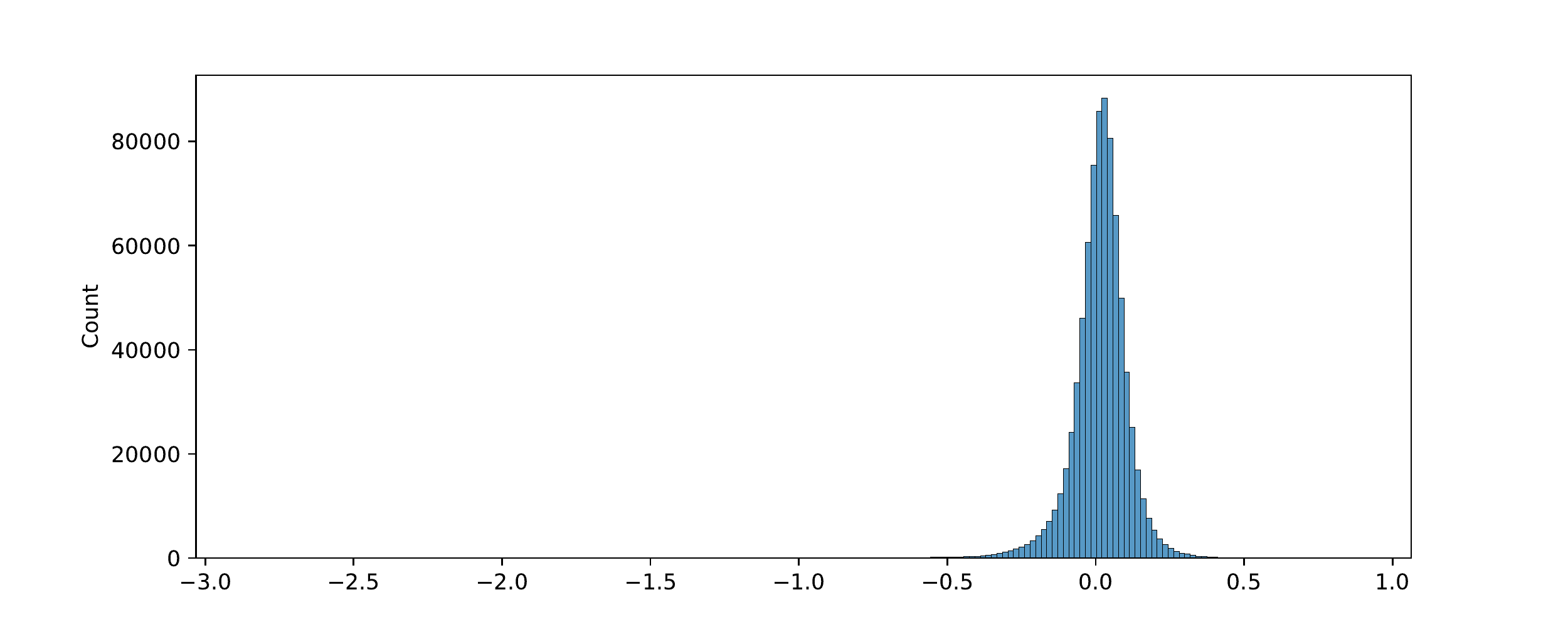} & \includegraphics[width=0.5\textwidth]{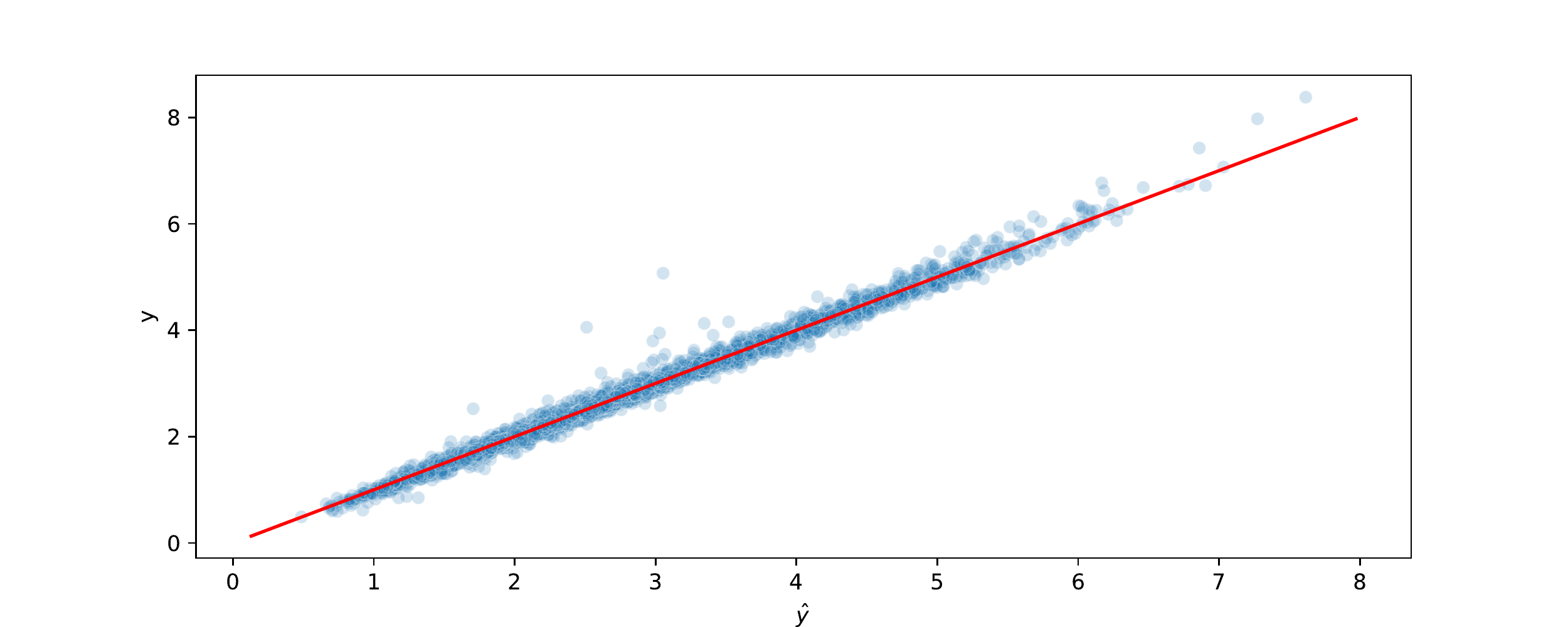}\\
	(a) Histogram of errors $\hat{y} - y$ & (b) $y$ vs $\hat{y}$\\
\end{tabular}
\caption{Comparison of out-of-sample predictions (50\% quantiles) $\hat y$ and observed drag coefficients $y$.}
\end{figure}

Finally, we show histograms of the posterior predictive distribution for four randomly chosen out-of-sample response values in Figure \ref{fig:satellite-hist}. We can see that the model concentrates the distribution of around the true values of the response.
\begin{figure}[H]\label{fig:satellite-hist}
\begin{tabular}{cc}
	\includegraphics[width=0.5\textwidth]{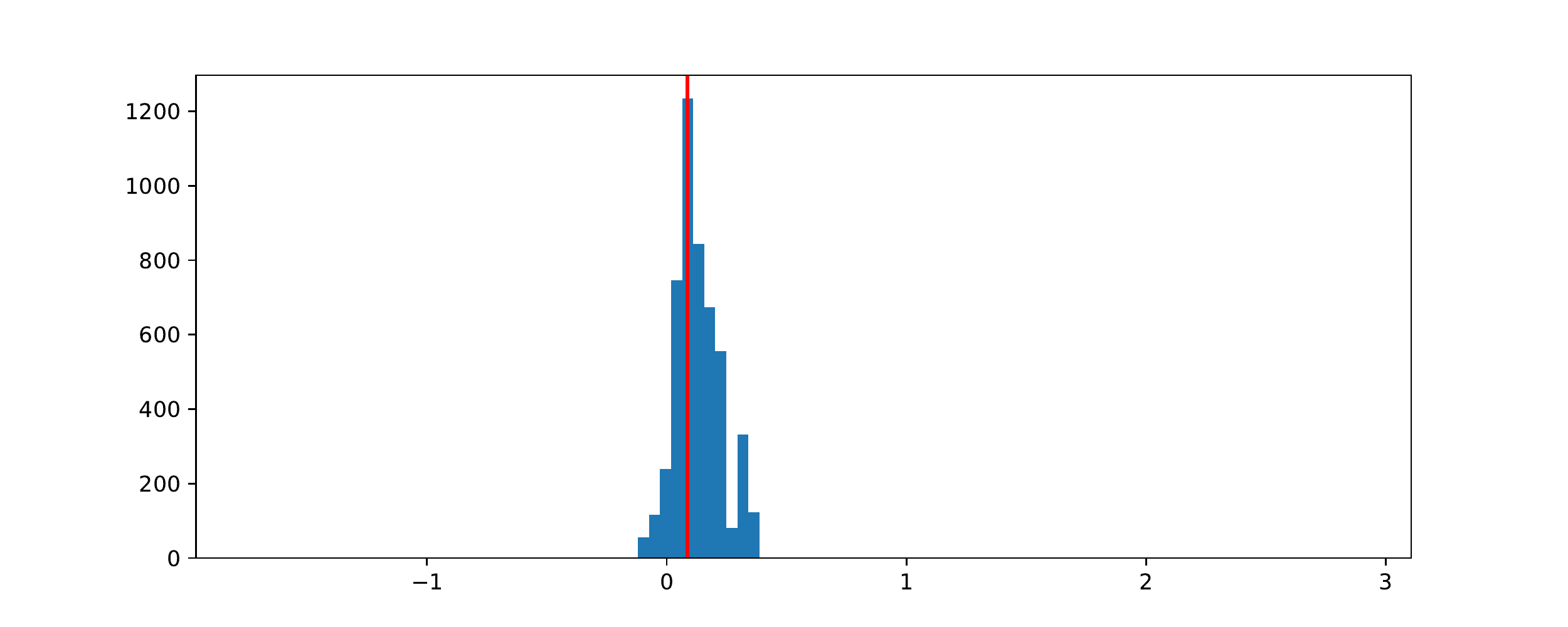} & \includegraphics[width=0.5\textwidth]{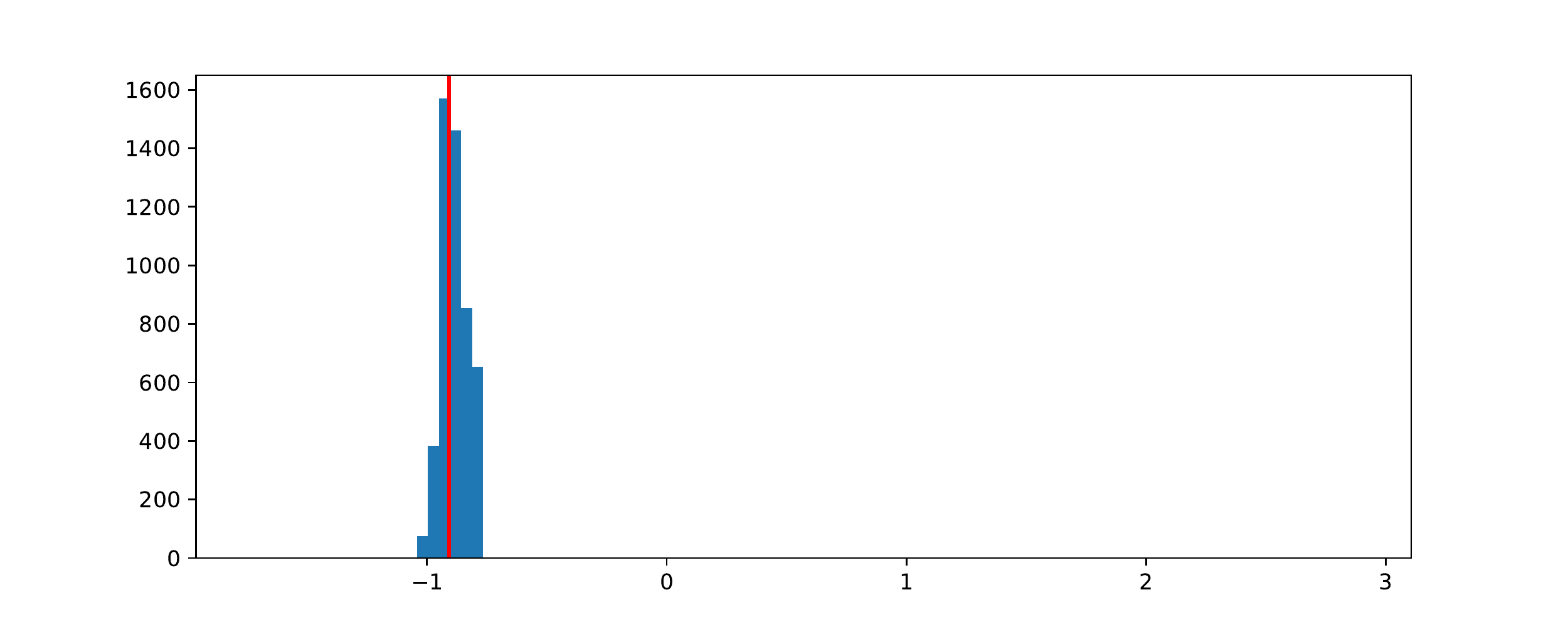}\\
	(a) Observation 948127 & (b) Observation 722309\\
	\includegraphics[width=0.5\textwidth]{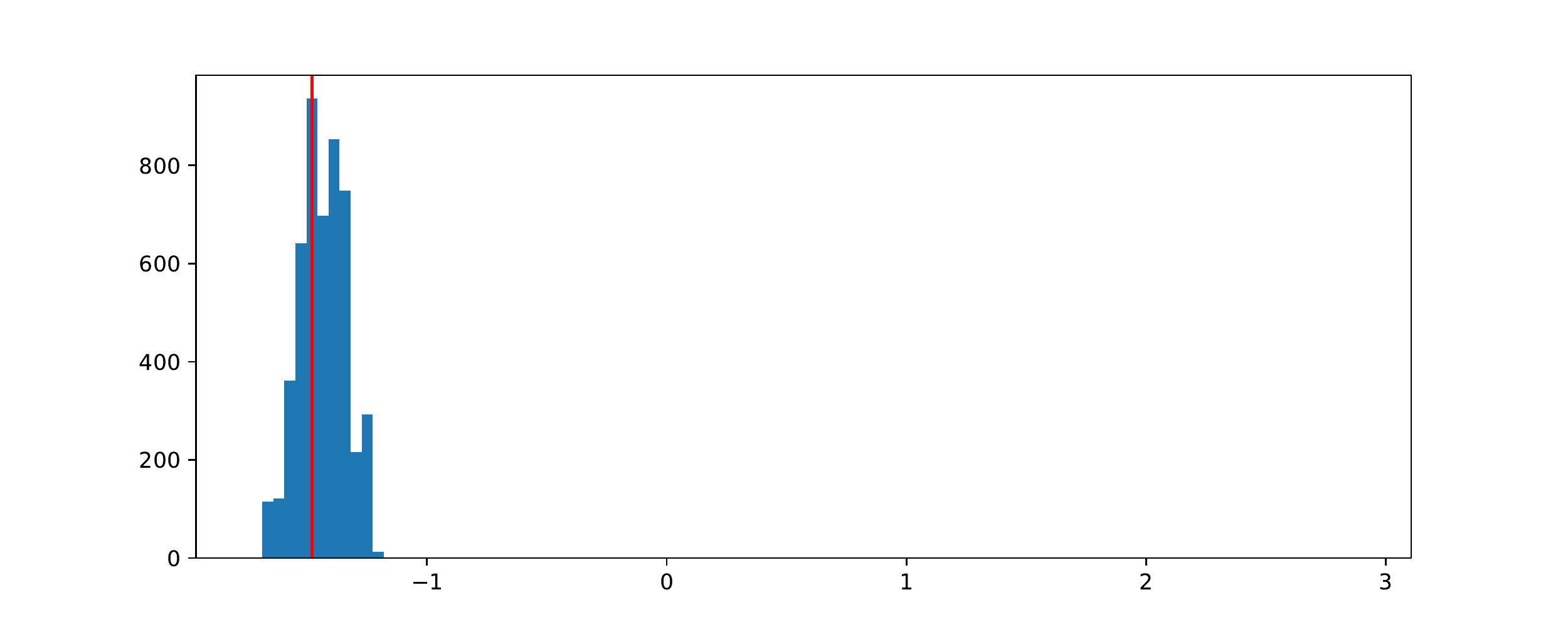} & \includegraphics[width=0.5\textwidth]{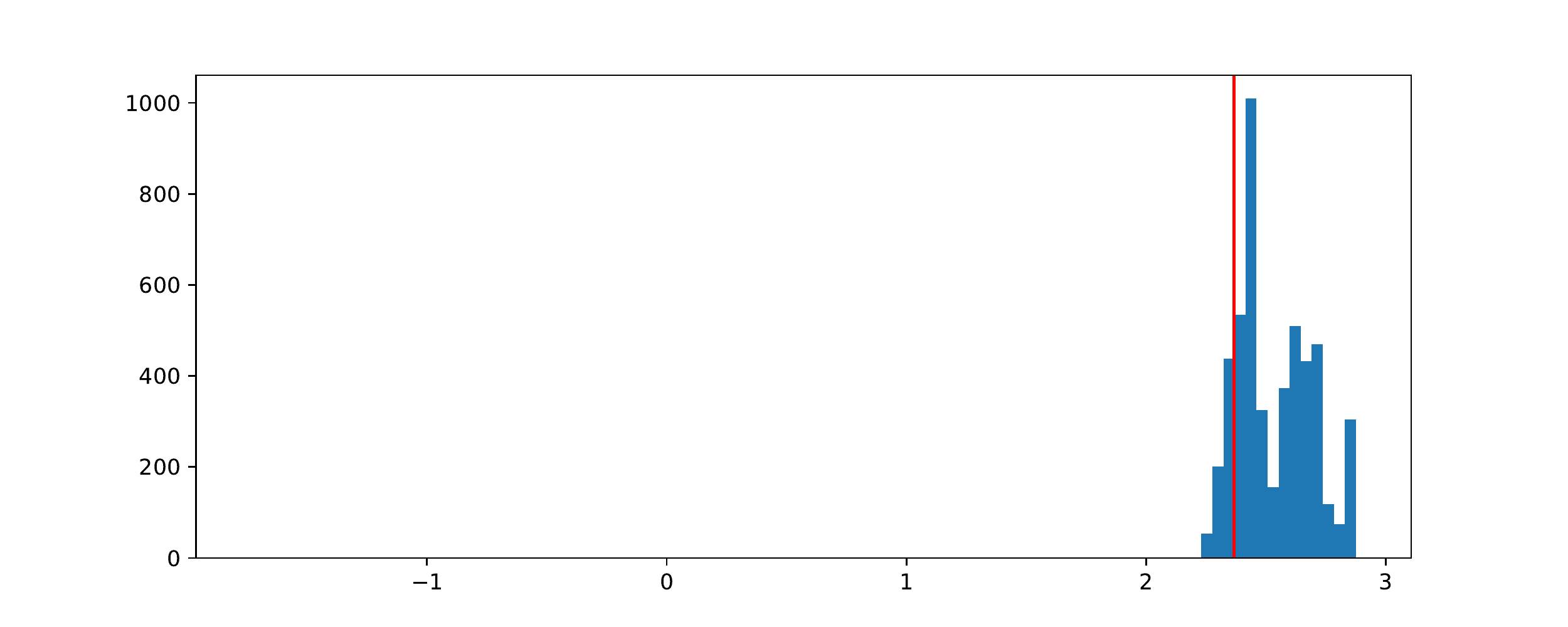}\\
	(a) Observation 608936 & (b) Observation 988391
\end{tabular}
\caption{Posterior histograms for four randomly chosen out-of-sample response values. The vertical red line is the observed value of the response.}
\end{figure}

Overall, our model provides a competitive performance to the state-of-the art techniques used for predicting and UQ analysis of complex models, such as satellite drag. The model is able to capture the distribution of the response and provide accurate predictions. The model is also able to provide uncertainty quantification in the form of credible prediction intervals.

\section{Discussion}\label{sec:discussion}
Generative AI is a simulation-based methodology, that takes joint samples of observables and parameters as an input and then applies nonparametric regression in a form of deep neural network by regressing $\theta$ on a non-linear function $h$ which is a function of dimensionality-reduced sufficient statistics of $\theta$ and a randomly generated stochastically uniform error. In its simplest form, $h$, can be identified with its inverse CDF. 

One solution to the multi-variate case is to use auto-regressive quantile neural networks. There are also many alternatives to the architecture design that we propose here. For example, autoencoders \cite{albert2022,akesson2021} or 
implicit models, see  \cite{diggle1984,baker2022,schultz2022}
There is also a link with indirect inference methods developed in \cite{pastorello2003,stroud2003nonlinear,drovandi2011,drovandi2015}

There are many challenging future problems. The method can easily handle high-dimensional latent variables. But designing the architecture for fixed high-dimensional parameters can be challenging. We leave it for the future research. Having learned the nonlinear map, when faced with the observed data $y_{\mathrm{obs}} $, one simply evaluates the nonlinear map at newly generated uniform random values. Generative AI circumvents the need for methods like MCMC that require the density evaluations. 

We also think that  over-parameterisation of the problem might  lead to a useful model, although it might also lead to non-identifiabaility of the weights in the regression. Two interesting approaches include,  the case when $K >k $ and mixture models with mixtures of Gaussian for $\tau$.  

\bibliographystyle{apalike}
\bibliography{InverseBayes,polsok,ref}
\end{document}